\documentclass[runningheads,letterpaper]{llncs}
\pdfoutput=1

\usepackage[ngerman,english]{babel}
\usepackage[latin1]{inputenc}
\usepackage[T1]{fontenc}
\usepackage[top=1.5in,bottom=1.5in,left=1.5in,right=1.5in]{geometry}
\usepackage{amsmath}                   % Für Formel-Blöcke, diverse
\usepackage{amssymb}                   % Mathe-Symbole (z.B. \subsetneq)
\usepackage{color}                     % Für Farben, eigtl. nur für \todo
\usepackage{comment}                   % Kommentar-Umgebungen
\usepackage{graphicx}
\usepackage{bm}

\usepackage{hyperref}

\hypersetup{%
  colorlinks=true,
  linkcolor=blue,
  anchorcolor=black,
  citecolor=blue,
  filecolor=blue,
  pdfpagemode=None,
  pdfstartview=FitH
}

\hypersetup{%
  pdftitle={Malicious Bayesian Congestion Games},
  pdfauthor={Martin Gairing},
  pdfsubject={arXiv Version},
  pdfkeywords={Game Theory, Selfish Routing}
}

\usepackage{colortbl}
\usepackage{hhline}
\definecolor{Gray}{gray}{0.8}
\definecolor{Lightgray}{gray}{0.9}
\definecolor{Darkgray}{gray}{0.5}


\newcommand{\cref}[1]{Corollary~\ref{#1}}

\newcommand{\G}{\Gamma}
\newcommand{\E}{[m]}

\newcommand{\N}{\ensuremath{\mathcal{N}}}

\newcommand{\Opt}{\ensuremath{\textsf{OPT}}}
\newcommand{\OPT}{\Opt}
\newcommand{\SC}{\ensuremath{\textsf{SC}}}
\newcommand{\PC}{\ensuremath{\textsf{PC}}}
\newcommand{\PoA}{\ensuremath{\textsf{PoA}}}
\newcommand{\poa}{\PoA}

\newcommand{\vect}[1]{\ensuremath{\boldsymbol{#1}}}

\renewcommand{\S}{\ensuremath{\mathcal{S}}}
\newcommand{\lsel}{\delta}
\newcommand{\lmal}{\kappa}

\newcommand{\s}{\vect s}
\newcommand{\Q}{\mathbf{Q}}

\newcommand{\NP}{\cal{NP}}
\newcommand{\xn}{\overline{x}}

\graphicspath{{Figures/}}

\allowdisplaybreaks[1]\allowdisplaybreaks[1]

\title{{\bf Malicious Bayesian Congestion Games}%
         \thanks{This work was supported by a fellowship within the Postdoc-Programme 
                        of the German Academic Exchange Service (DAAD).}
}
\author{Martin Gairing}
\institute{International Computer Science Institute,
           Berkeley, CA, USA.\\
           \email{gairing@icsi.berkeley.edu}
}

\begin{document}
\maketitle

\newcommand{\games}{\ensuremath{\mathcal{G}}}

\newcommand{\pc}{\PC}
\newcommand{\PoB}{\ensuremath{\textsf{PoB}}}
\newcommand{\pob}{\PoB}
\newcommand{\PoM}{\ensuremath{\textsf{PoM}}}
\newcommand{\pom}{\PoM}
\newcommand{\WoM}{\ensuremath{\textsf{WoM}}}
\newcommand{\wom}{\WoM}

\newcommand{\ntf}{\ensuremath{\left\lfloor \frac{n}{2} \right\rfloor}}
\newcommand{\ntc}{\ensuremath{\left\lceil \frac{n}{2} \right\rceil}}
\newcommand{\nmf}{\ensuremath{\left\lfloor \frac{n}{r} \right\rfloor}}
\newcommand{\nmc}{\ensuremath{\left\lceil \frac{n}{r} \right\rceil}}

\newcommand{\tsl}{\Delta}

\newcommand{\p}{\vect p}
\newcommand{\vsi}{\ensuremath{\bm{\sigma}}}
%\newcommand{\vmu}{\ensuremath{\bm{\mu}}}

%\renewcommand{\P}{\mathbf{P}}
%\newcommand{\s}{\vect s}
%\newcommand{\Q}{\mathbf{Q}}
%\newcommand{\support}{\ensuremath{\textsf{support}}}
%\newcommand{\Cj}[1]{\ensuremath{\sum_{j\in #1}\frac{1}{a_j}}}

%\newcommand{\la}{\lambda}

%\newcommand{\sumin}{\sum_{i\in\N}}

%%%%%%%%%%%%%%%%%%%%%%%%%%%%%%%%%%%%%%%%%%%%%%%%%%%%%%%%%%%%%%% new
%\newcommand{\NP}{\cal{NP}}
%\newcommand{\xn}{\overline{x}}

%%\graphicspath{{Figures/}}

%\allowdisplaybreaks[1]\allowdisplaybreaks[1]

%\title{Malicious Bayesian Congestion Games%
%         \footnote{This work has been supported
%	         by the German Accademic Exchange Program.}
%}

%\author{Martin Gairing}
%\authorrunning{M.~Gairing}

%\institute{International Computer Science Institute,
%           Berkeley, CA, USA.\\
%           \email{gairing@icsi.berkeley.edu}
%}

%
%\date{\today}

%\begin{document}
%%
%\maketitle

%
%%%%%%%%%%%%%%%%%%%%%%%%%%%%%%%%%%%%%%%%%%%%%%%%%%%%%%%%%%%%%%% ABSTRACT
\begin{abstract}%\normalsize
In this paper, we 
introduce \emph{malicious Bayesian congestion games} as an extension to
congestion games where players might act in a malicious way. 
In such a game each player has two \emph{types}. 
Either the player
is a rational player seeking to minimize her own delay, or -- with a certain probability -- the player
is \emph{malicious} in which case her only goal is to disturb the other players as much as possible.  

We show that such games do in general not possess a Bayesian Nash equilibrium in pure strategies
(i.e. a \emph{pure Bayesian Nash equilibrium}). 
Moreover, given a game, we show that it is \NP-complete
to decide  whether it admits a pure Bayesian Nash equilibrium.
This result even holds when resource latency 
functions are linear, each player is malicious with the same probability, and all strategy sets consist of
singleton sets of resources.
For a slightly more restricted class of malicious Bayesian congestion games, we provide easy checkable properties
that are necessary and sufficient for the existence of a pure Bayesian Nash equilibrium.

In the second part of the paper we study the impact of the malicious types on the overall performance 
of the system (i.e. the \emph{social cost}). To measure this impact, we use the \emph{Price of Malice}. 
We provide (tight) bounds on the Price of Malice for 
an interesting class of malicious Bayesian congestion games. Moreover, we show that for 
certain congestion games the advent of malicious types can also be beneficial 
to the system in the sense that the social cost of the worst case equilibrium decreases. We provide 
a tight bound on the maximum factor by which this happens.
\end{abstract}
%
%
%
%%%%%%%%%%%%%%%%%%%%%%%%%%%%%%%%%%%%%%%%%%%%%%%%%%%%%%%%%%%%%%% INTRODUCTION
%
\section{Introduction}
\label{s:intro}
\textbf{Motivation and Framework.}
Over the last decade, the study of strategic behavior in distributed systems 
has improved our understanding of modern computer artifacts such as the Internet.
Normally, the users of such distributed systems are modeled as rational, utility optimizing players. 
However, in many real world scenarios, users do not necessarily act rational, but rather \emph{irrational}. 
In this paper, we address one form of irrationality, namely, we allow that players act  in a \emph{malicious} way. In this case, the only goal of a malicious player is to disturb the (non-malicious) 
players as much as possible. 
The presence of \emph{Denial of Service attacks} in the Internet is an example showing 
that such systems are quite realistic.
In many such systems with malicious players, the players have 
only \emph{incomplete information} about the set of 
malicious players. 
A standard approach for modeling games with incomplete information
uses the Harsanyi transformation~\cite{Har67}, which converts a  game with incomplete 
information to a
game where players have different \emph{types}. The type of a player represents its private information
that is not common knowledge to all players. In the resulting \emph{Bayesian game}, each player's
uncertainty about each other's type is described by a probability distribution.

One aspect of Game Theory that was studied extensively in recent years 
is the {\em Price of Anarchy} as introduced by Koutsoupias
and Papadimitriou~\cite{KP99}. The Price of Anarchy is the worst case ratio between the value of the 
{\em social cost} in an equilibrium state of the system and that of some social optimum.
Usually, the equilibrium state is defined as {\em Nash equilibrium} -- a state in which no player can 
unilaterally improve her private objective function, also coined as {\em private cost}. 
A Nash equilibrium is {\em pure}  if all players choose a pure strategy and {\em mixed} if players
choose probability distributions over pure strategies.

%Another aspect extensively studied in the last years is the complexity of 
%computing mixed Nash equilibria. 
%Through a sequence of papers \cite{GP06,DGP06,CD06} it was shown that computing
% a mixed Nash equilibrium for 
%a finite game given in strategic form is PPAD-complete, even for two players \cite{CD06}.

While the celebrated result of Nash~\cite{Nas51} guarantees the existence of a mixed Nash equilibrium 
for ever finite game, pure Nash equilibria are not guaranteed to exists (see e.g. \cite{FKS05,GMV05,LO01,Mil96}). 
An natural question to ask, is whether a given game possesses a pure Nash equilibrium or not.
We address this question by asking about the complexity of this decision problem.

A class of games that always possess pure Nash equilibria is the class of congestion games as 
introduced by Rosenthal~\cite{Ros73a}. Here, the strategy set of each player is a subset of the power set of
given resources, the latency on each resource is described by a latency function in the number of
players sharing this resource, and the private cost of each player is the sum of the latencies of its
chosen resources. 
%In a congestion game with singleton strategy sets, each player strategy
%consists only of a single resource.
Milchtaich~\cite{Mil96} considered weighted congestion games as an extension to
congestion games in which the players have weights and thus different influence on the
latency of the resources. 

To measure the influence of malicious behavior, Moscibroda et al.~\cite{MSW06} introduced the 
{\em Price of Byzantine Anarchy} as the worst case ratio between the social cost in an equilibrium state
of the system under some assumption on the malicious players and the social cost of some social 
optimum without malicious players. 
They further define the {\em Price of Malice} as the ratio between the Price of 
Byzantine Anarchy and the Price
of Anarchy. We will use a similar definition and define the equilibrium state 
as a Bayesian Nash equilibrium. 

\medskip\noindent
\textbf{Contribution.}
In this paper, we 
introduce \emph{malicious Bayesian congestion games} as an extension to
congestion games where players might act in a malicious way. 
Following  Har\-sanyi's transformation \cite{Har67}, we  allow each player to be of two \emph{types}. 
Either the player
is a rational player seeking to minimize her own delay, or -- with a certain probability -- the player
is \emph{malicious} in which case her only goal is to disturb the other players as much as possible.  
For such games we study the complexity of deciding whether a given game has a pure 
Bayesian Nash equilibrium. Moreover, we study the impact of the malicious types on the overall performance 
of the system (i.e. the \emph{social cost}). To measure this impact, we use the \emph{Price of Malice},
which we define similarly as Moscibroda et al.~\cite{MSW06}. 

We now describe our results in more detail.
As our main result, we show that it is \NP-complete
to decide  whether a given malicious Bayesian congestion game admits a pure Bayesian Nash equilibrium.
This result even holds if resource latency 
functions are linear, each player is malicious with the same probability, and all strategy sets consist 
singleton sets (Theorem~\ref{t:npcscg}). The same result even holds if we further restrict to the case that
each player has at most four strategies and at most three players can be assigned to each resource 
(Theorem~\ref{t:npcscg_ex}).
For {\em symmetric} Bayesian congestion games with identical type probability, 
identical latency functions and
strategy sets that consist only of singletons, we provide easy checkable properties
that are necessary and sufficient for the existence of a pure Bayesian Nash equilibrium
(Theorem~\ref{t:characterization}).

We then shift gears and present results related to the Price of Malice. 
For general malicious Bayesian congestion games with linear latency functions, 
we show an upper bound
on the the Price of Byzantine Anarchy (Theorem~\ref{t:pob}).  
Moreover, we proof a lower bound on the same ratio that already holds for the case of  identical
type probabilities (Theorem~\ref{t:lbi}). 
As a corollary, we get an asymptotic tight bound on the Price of Malice (Corollary~\ref{c:pom}).
We close the paper with a tight lower bound on the maximum factor by which the social cost of a worst 
case (Bayesian) Nash equilibrium of a congestion game might decrease by introducing
 malicious types (Theorem~\ref{t:wom}). 

\medskip\noindent
\textbf{Related Work.}
Congestion games and variants thereof have long been used to model non-cooperative resource
sharing among selfish players. Rosenthal~\cite{Ros73a} showed that congestion games always possess
pure Nash equilibria. The complexity of computing such a pure Nash equilibrium has been settled 
for arbitrary latency functions by Fabrikant et al.~ \cite{FPT04} and later for linear latency functions by Ackermann et al.~\cite{ARV06a}. 
On the other hand, for weighted congestion games, Libman and Orda~\cite{LO01},
Fotakis et al.~\cite{FKS05} and Goemans et al.~\cite{GMV05} provide examples that do not 
allow for a pure
Nash equilibrium. 
Dunkel and Schulz~\cite{DS06} showed that it is \NP-complete to decide the existence 
of a pure Nash equilibrium for a given weighted congestion games. 

The Price of Anarchy for weighted congestion games has been studied extensively
(see e.g. \cite{AAE05,ADGMS06,CK05}). In case of linear latency functions, 
the Price of Anarchy is exactly $\frac{5}{2}$ for unweighted congestion games \cite{CK05} and
$1+\Phi$ for weighted congestion games \cite{AAE05}, where $\Phi=\frac{1+\sqrt{5}}{2}$ is the golden ratio.
The exact value of the Price of Anarchy is also known for the case of polynomial latency functions \cite{ADGMS06}.
For bounds on the Price of Anarchy of (weighted) congestion games with each strategy set being 
a singleton set of resources, we refer to
\cite{GS07} and references therein.

Several recent papers considered games allowing for malicious player behavior  
\cite{BKP07,KV03,MSW06}. 
Moscibroda et al.~\cite{MSW06} introduced the Price of Malice and gave bounds on the Price of Malice for 
a virus inoculation game where some of the players are malicious. In fact, our definition of Price of Malice
is motivated by the corresponding definition from this paper.  
Karakostas et al.~\cite{KV03} and Babaioff et al.~\cite{BKP07},  study malicious player behavior 
in {\em non-atomic} congestion games. Here, each player from a continuum of infinitely many players controls only an infinitesimal small amount of weight and a fraction of those players is malicious.  
In contrast to those papers, our games are atomic, and thus have only finitely many players. 
This yields to different results.

For general Bayesian games, questions concerning the complexity of 
deciding the existence of a pure Bayesian Nash equilibrium 
have been addressed in two recent works \cite{CS03,GGM07}. 
On the one hand, 
if the game is given in \emph{standard normal form}, i.e. the utility functions and the type probability
distribution are represented extensively as tables, then deciding the existence of a pure Bayesian
Nash equilibrium is \NP-complete \cite{CS03}.
On the other hand, if both -- the utility functions and the type probability distribution -- 
are succinctly encoded, then the problem becomes PP-complete \cite{GGM07}.  
%, even if players can assume only two types 
In contrast to \cite{CS03}, malicious Bayesian congestion games are succinctly represented but they are
more structured as the games considered by Gottlob et al.~\cite{GGM07}.

A certain class  of Bayesian congestion game
has been introduced in 
\cite{GMT08}. 
Here, players act completely rational but they are uncertain about each others weight.
Among other results, the authors show that such games always possess pure Bayesian Nash equilibria
if latency functions are linear.

%\textbf{Motivation and Framework.}
%
%
%%%%%%%%%%%%%%%%%%%%%%%%%%%%%%%%%%%%%%%%%%%%%%%%%%%%%%%%%%%%%%% NOTATION
%\section{Model}\label{s:notation}
%\noindent
%\textbf{General.}

%%%%%%%%%%%%%%%%%%%%%%%%%%%%%%%%%%%%%%%%%%%%%%%%%%%%%%%%%%%%%%% UPPER BOUND
\section{Model}
\label{s:model}
\subsection{Congestion Games}
\label{s:cg}

\textbf{Instance.}
A {\em congestion game} $\Gamma$ is a tuple
$$
  \Gamma = \left(\N,E,(S_u)_{u\in\N},(f_e)_{e\in E}\right).
$$
Here, $\N$ is the set of {\em players} and $E$ is the
finite set of
{\em resources}. 
Throughout, we
denote $n=|\N|$ and $r=|E|$ and assume $n\ge2$ and $r\ge 2$.
For every player $u\in\N$, $S_u\subseteq 2^E$ is the {\em strategy set} of player $u$.
Denote 
$S=S_1\times\ldots\times S_n$. 
%and
%$S_{-i}=S_1\times\ldots\times S_{i-1}\times S_{i+1} \ldots\times S_n$.
For every 
resource $e\in E$, the {\em latency function}
$f_e: \mathbb{N} \rightarrow \mathbb{R}$
is a non-negative, non-decreasing function that describes the 
{\em latency} on resource $e$. 
For most of our results, we consider {\em affine latency functions} with 
non-negative coefficients, that is, for all resources
$e\in E$, the latency function is of the form
$
  f_e(\delta)= a_e \cdot \delta + b_e
$
with $a_e, b_e \geq 0$. Affine latency functions are \emph{linear} if
$b_e=0$ for all $e\in E$.
A congestion game is called {\em symmetric}, if $S_u=S_u'$ for
any pair of players $u,u'$.

\medskip\noindent
\textbf{Strategies and Strategy Profiles.}
A {\em pure strategy} for player $u$ is some specific strategy $s_u\in S_u$, while a
{\em mixed strategy} $Q_u = (q(u,s_u))_{s_u \in S_u}$ is a
 a probability distribution over $S_u$, where 
 $q(u,s_u)$ denotes
the probability that player $u$ chooses the pure strategy $s_u$.

A {\em pure strategy profile} is an $n$-tuple $\s=(s_1,\ldots, s_n)$
whereas a {\em mixed strategy profile} $\Q=(Q_1,\ldots,Q_n)$ is represented
by an $n$-tuple of mixed strategies.
For a mixed strategy profile $\Q$, denote by
$$
   q(\s) = \prod_{u\in\N} q(u,s_u)
$$
 the probability that the players choose the pure strategy profile 
$\s$.

\medskip\noindent
\textbf{Load and Private Cost.}
For a pure strategy profile $\s$, denote by  
$\delta_e(\s) = \sum_{u\in\N: e\in s_u} 1$ the
{\em load} on resource $e\in\E$, i.e. the number of players assigned to $e$.
In the same way, for a partial strategy profile $s_{-i}$,
denote
$\delta_e(\s_{-i}) = \sum_{u\in\N, u\neq i: e\in s_u} 1$ the
{\em load} on resource $e\in\E$ without player $i$.

Fix a pure strategy profile
$\s$.
The {\em private cost} $\pc_u(\s)$ of player
$u\in\N$ is defined by the {\em latency} of the chosen resources. Thus 
$$
  \pc_u(\s)
  = \sum_{e\in s_u}f_e 
           \left( \delta_e(\s)\right).
$$
For a mixed strategy profile $\Q$, the {\em private cost} of 
player $u\in\N$ is
$$
  \pc_u(\Q)
  = \sum_{\s \in S} q(\s) \cdot \pc_u(\s)\; .
$$

\medskip\noindent
\textbf{Social Cost.}
Associated with a
congestion game $\Gamma$ and a pure strategy profile $\s$ is the
{\em social cost} $\SC(\G,\s)$ as a measure of social welfare.
In particular we use the expected average latency. That is,
\begin{eqnarray*}
\SC(\G,\s) &=&  \frac{1}{n}
             \sum_{e \in E}
               \delta_e(\s) \cdot f_e(\delta_e(\s))\\
&=& \frac{1}{n} \sum_{\s\in S} q(\s)
    \sum_{u\in\N}
    \sum_{e \in s_u}
      f_e(\delta_e(\s))\\
&=& \frac{1}{n} \sum_{u\in\N}  \pc_u(\Q).
\end{eqnarray*}
Observe, that this measure differs from the {\em total latency} \cite{RT02} only by the factor $n$.

The {\em optimum}
associated with a congestion game $\G$ is the least possible social cost,
over all pure strategy profiles $\s \in S$. Thus,
$$
  \OPT(\G)= \min_{\s\in S} {\SC}(\G,\s)\, .
$$

\medskip\noindent
\textbf{Nash Equilibria.}
Given a
congestion game and an associated mixed strategy profile
$\Q$, player $u \in \N$ is {\em satisfied} if the player cannot
improve its private cost by unilaterally changing its strategy.
Otherwise, player $u$ is {\em unsatisfied}. The mixed strategy
profile $\Q$ is a {\em Nash equilibrium} if and only if all
players $u \in \N$ are satisfied, that is,
$
  \pc_u(\Q) \leq \pc_u(\Q_{-u},s_u)
$
for all $u\in\N$ and $s_u\in S_u$.

%Note, that if this inequality holds for all pure strategies $s_i\in S_i$
%of player $i$, then it also holds for all mixed strategies over $S_i$.
Depending on the type of strategy profile we differ between
{\em pure} and {\em mixed} Nash equilibria.

\medskip\noindent
\textbf{Price of Anarchy.}
Let $\games$ be a class of congestion games.
The {\em Price of Anarchy}, also called {\em coordination ratio}
and denoted by $\poa$, is the supremum, over all instances
$\Gamma\in \games$ and Nash equilibria $\Q$, of the ratio
$\frac{\SC(\G,\Q)}{\OPT(\G)}$. Thus,
\begin{align*}
\poa & = \sup_{\Gamma\in\games, \Q} \frac{\SC(\G,\Q)}{\OPT(\G)}.
\end{align*}

\subsection{Malicious Bayesian Congestion Games}
\label{s:mbcg}
%\medskip
\noindent
\textbf{Instance.}
A {\em malicious Bayesian congestion game} $\Psi$ is an extension to congestion games,
where each player is malicious with a certain probability.
Following Harsanyi's approach, we model such a game with incomplete information
as a Bayesian game, where each player $u\in\N$ can be of two types:
Either $u$ is \emph{selfish} or \emph{malicious}.
For each type of player $u\in\N$ we introduce two independent 
type-agents $u^s$ and $u^m$, denoting the {\em selfish} and {\em malicious type-agent} 
of player $u$, respectively.

Let $p_u$ be the probability that player $u\in\N$ is malicious and call $p_u$ the \emph{type probability}
of player $u$.
Define the \emph{type probability vector} $\p=(p_1, \ldots, p_n)$
in the natural way. Denote $p_{\min}=\min_{u\in\N} p_u$.
In the case of identical type probabilities $p_u=p$ for all player $u\in\N$.
Define $\tsl=\sum_{u\in\N} p_u$ as the \emph{expected number of malicious players}. 
Observe, that for identical type probabilities $\tsl=p\cdot n$.
Denote by $\G_\Psi$ the congestion game that arises from the malicious Bayesian congestion game 
$\Psi$ by setting $p_u=0$ for all player $u\in\N$.

Summing up, a malicious Bayesian congestion game $\Psi$ is given by a tuple 
$$
   \Psi=\left(\N,E,(S_u)_{u\in\N}, (p_u)_{u\in\N}, (f_e)_{e\in E}\right).
$$

\medskip\noindent
\textbf{Strategies and Strategy Profiles.}
A pure strategy $\sigma_u$ for player $u\in \N$ is now a tuple  
$\sigma_u=(\sigma(u^s),\sigma(u^m)) \in S_u^2$, where $\sigma(u^s)$ and $\sigma(u^m)$
denote the strategy of the selfish type-agent and malicious type-agent of player $u$, 
respectively.
Denote $\vsi=(\sigma_1,\ldots,\sigma_n)$.
A {\em mixed strategy} $Q_i$ is now 
a probability distribution over $S_i\times S_i$. Define $\Q$ and $q(\vsi)$
as before.

\medskip\noindent
\textbf{Private Cost.}
For any type probability vector $\p$ and pure strategy profile
$\vsi$, 
denote the
{\em expected selfish load}  on resource $e\in E$ by 
$\lsel_e(\vsi) = \sum_{u\in\N: e\in \sigma(u^s)} (1-p_u)$ and  the
{\em expected malicious load} by
$\lmal_e(\vsi) = \sum_{u\in\N: e\in \sigma(u^m)} p_u$.
For a partial assignment $\vsi_{-u}$ define
$\lsel_e(\vsi_{-u})$ and $\lmal_e(\vsi_{-u})$ accordingly, by disregarding player $u$.

Fix any type probability vector $\p$ and pure strategy profile
$\vsi$.
The {\em private cost} $\pc_u(\p,\vsi)$ of player
$u\in\N$ is defined by
$$
  \pc_u(\p,\vsi)
  = \sum_{e\in \sigma(u^s)}f_e 
           \left( \lsel_e(\vsi_{-u}) + \lmal_e(\vsi_{-u}) + 1 \right).
$$
In other words $\pc_u(\p,\vsi)$ is the expected latency that player $u$ 
experiences if player  $u$ is selfish. For each player $u\in\N$, 
type-agent $u^s$ aims to minimize $\pc_u(\p,\vsi)$.
Observe, that  $\pc_u(\p,\vsi)$ does not depend on $\sigma(u^m)$.
For a mixed strategy profile $\Q$,  define
$
  \pc_u(\p,\Q)
$
accordingly.

\medskip\noindent
\textbf{Social Cost.}
%We generalize the definition of social cost as follows:
Let $\Psi$ be a malicious Bayesian congestion game with type probability vector $\p$ 
and let $\Q$ be a pure strategy profile for $\Psi$.
We generalize the definition of 
{\em social cost} $\SC(\Psi,\Q)$ to the  
weighted average latency of the selfish type-agents. 
That is,
\begin{eqnarray*}
\SC(\Psi,\Q) 
&=& \frac{\sum_{u\in\N} (1-p_u) \cdot \pc_u(\p,\Q)}
                 {n-\tsl}.
\end{eqnarray*}

\medskip\noindent
\textbf{Bayesian Nash equilibria.}
A selfish type-agent is {\em satisfied} if she cannot unilaterally decrease her private cost,  that is,
$
  \pc_u(\Q) \leq \pc_u(\Q_{-u^s}, \sigma(u^s))
$
for all $u\in\N$ and $\sigma(u^s)\in S_u$.

In contrast to the selfish type-agents, each malicious type-agent aims to maximize social cost.
So, a malicious type-agent is {\em satisfied} if she cannot increase social cost by unilaterally changing
her strategy. 
%For a mixed strategy profile $\Q$, one can easily show that a malicious type-agent $u^m$ is 
%satisfied if and only if she chooses
%only strategies $\sigma(u^m)\in S_u$ that maximize
%$$
%   some relation
%$$

For a malicious Bayesian congestion game, a mixed strategy
profile $\Q$ is a {\em Bayesian Nash equilibrium} if and only if both type-agents 
of all
players $u \in \N$ are satisfied. 
Depending on the type of strategy profile we again differ between
{\em pure} and {\em mixed} Bayesian Nash equilibria.

\medskip\noindent
\textbf{Price of Byzantine Anarchy and Price of Malice.}
For a fixed expected number of malicious players $\tsl$,
let $\games(\tsl)$ be the class of malicious Bayesian congestion games where
$\sum_{u\in\N} p_u = \tsl$.
Similarly to \cite{MSW06}, we define the
{\em Price of Byzantine Anarchy}, denoted by $\pob$, as the supremum, over all instances
$\Psi\in\games(\tsl)$ and Bayesian Nash equilibria $\Q$, of the ratio
between the social cost in $\Q$ and the optimum social cost of the corresponding congestion game
$\Gamma_\Psi$. Thus,
\begin{align*}
\pob(\tsl) & = \sup_{\Psi\in\games(\tsl), \Q} \frac{\SC(\Psi,\Q)}{\OPT(\Gamma_\Psi)}.
\end{align*}
Observe that for $\tsl=0$, the Price of Byzantine Anarchy $\pob(0)$ reduces to the 
Price of Anarchy $\poa$ as defined in Section~\ref{s:cg}.

Again similarly to \cite{MSW06}, we define the {\em Price of Malice} by
$$
   \pom(\tsl)=\frac{\pob(\tsl)}{\pob(0)}.
$$

\section{Existence and Complexity of pure Bayesian Nash equilibria}
\label{s:complexity}

In this section, we study the complexity of deciding whether a given malicious Bayesian congestion game
possesses a pure Bayesian Nash equilibrium or not. 

\begin{theorem}
\label{t:npcscg}
The problem of deciding whether a malicious Bayesian congestion 
game with linear latency functions possesses a pure Bayesian Nash 
equilibrium is \NP-complete, 
even if all strategy sets consist of singletons and either of the following
properties holds:
\begin{enumerate}
\item[(a)]
All players are malicious with the same probability $p$ for any $0<p<1$.
\item[(b)]
Only one player is malicious with positive probability $p$ for any $0<p\le1$.
\end{enumerate} 
\end{theorem}

\begin{proof}
Our proof uses a reduction from a restricted version of \textsc{3-SAT}.
Here, \textsc{3-SAT} is restricted to instances where
each clause is a disjunction of 2 or 3 variables and each variable occurs 
at most three times. Tovey \cite{Tov84} showed that it is \NP-complete to 
decide the satisfiability of such instances. 
Consider an arbitrary instance of 
\textsc{3-SAT} with set of variables $X=\{x_1, \ldots x_\ell\}$ and
set of clauses $C=\{c_1, \ldots c_k\}$.
Without loss of generality, we may assume that each variable occurs at most twice 
unnegated and at most twice negated.

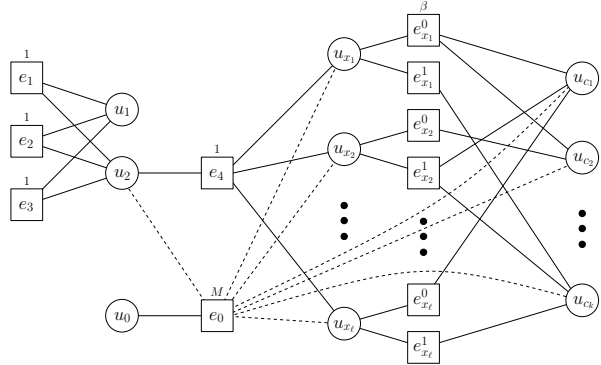
\begin{figure*}[ht]
\begin{center}
\makeatletter
\resizebox*{0.6\textwidth}{!}{%
\begin{picture}(0,0)%
\includegraphics{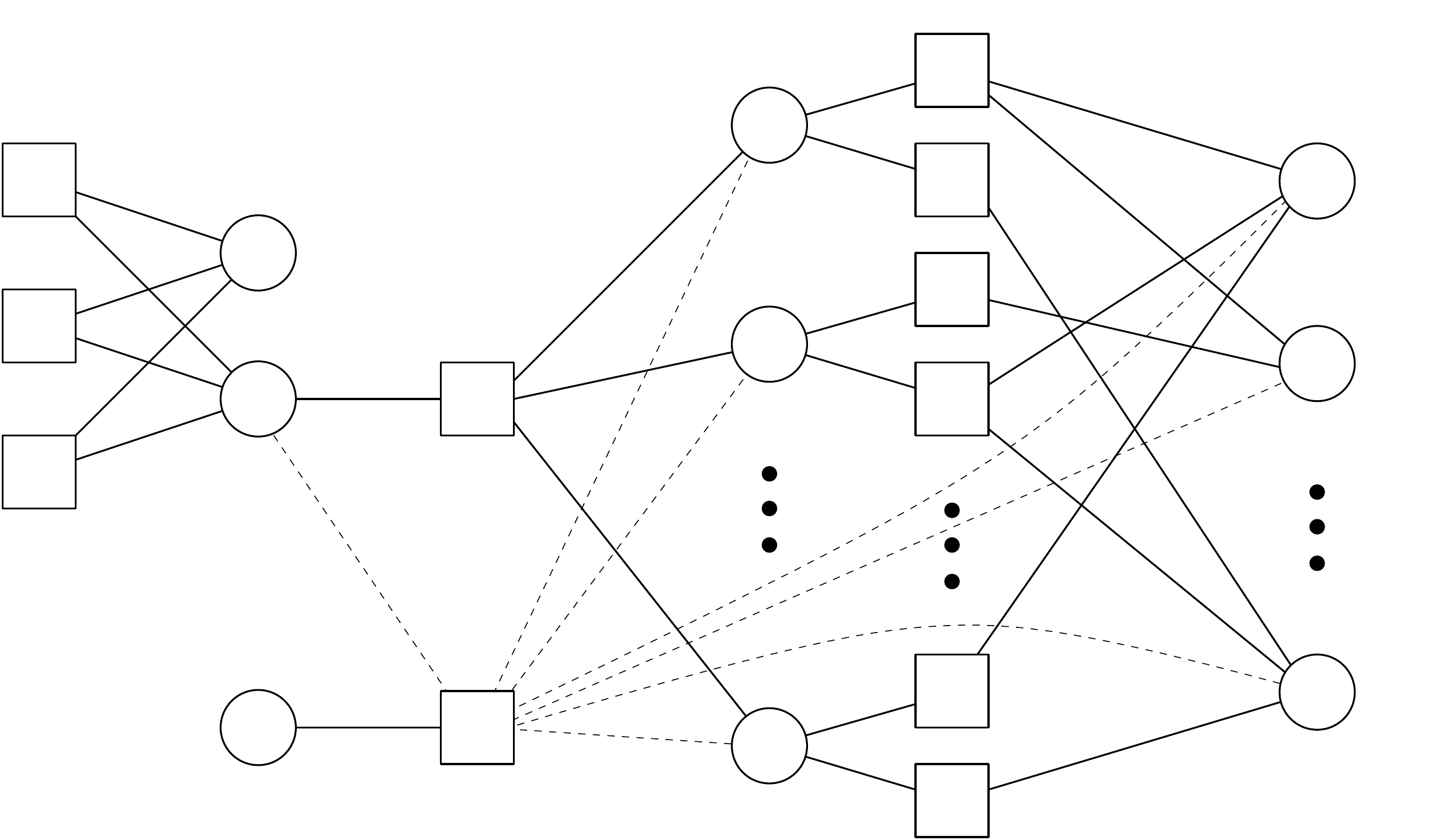}%
\end{picture}%
\setlength{\unitlength}{3947sp}%
\begingroup\makeatletter\ifx\SetFigFont\undefined%
\gdef\SetFigFont#1#2#3#4#5{%
  \reset@font\fontsize{#1}{#2pt}%
  \fontfamily{#3}\fontseries{#4}\fontshape{#5}%
  \selectfont}%
\fi\endgroup%
\begin{picture}(11918,6889)(4479,-6383)
\put(10596,-586){\makebox(0,0)[lb]{\smash{{\SetFigFont{20}{24.0}{\familydefault}{\mddefault}{\updefault}{\color[rgb]{0,0,0}$u_{x_1}$}%
}}}}
\put(10596,-2386){\makebox(0,0)[lb]{\smash{{\SetFigFont{20}{24.0}{\familydefault}{\mddefault}{\updefault}{\color[rgb]{0,0,0}$u_{x_2}$}%
}}}}
\put(10596,-5686){\makebox(0,0)[lb]{\smash{{\SetFigFont{20}{24.0}{\familydefault}{\mddefault}{\updefault}{\color[rgb]{0,0,0}$u_{x_\ell}$}%
}}}}
\put(12106,-136){\makebox(0,0)[lb]{\smash{{\SetFigFont{20}{24.0}{\familydefault}{\mddefault}{\updefault}{\color[rgb]{0,0,0}$e_{x_1}^0$}%
}}}}
\put(12106,-1036){\makebox(0,0)[lb]{\smash{{\SetFigFont{20}{24.0}{\familydefault}{\mddefault}{\updefault}{\color[rgb]{0,0,0}$e_{x_1}^1$}%
}}}}
\put(12106,-1936){\makebox(0,0)[lb]{\smash{{\SetFigFont{20}{24.0}{\familydefault}{\mddefault}{\updefault}{\color[rgb]{0,0,0}$e_{x_2}^0$}%
}}}}
\put(12106,-2836){\makebox(0,0)[lb]{\smash{{\SetFigFont{20}{24.0}{\familydefault}{\mddefault}{\updefault}{\color[rgb]{0,0,0}$e_{x_2}^1$}%
}}}}
\put(12106,-5236){\makebox(0,0)[lb]{\smash{{\SetFigFont{20}{24.0}{\familydefault}{\mddefault}{\updefault}{\color[rgb]{0,0,0}$e_{x_\ell}^0$}%
}}}}
\put(12106,-6136){\makebox(0,0)[lb]{\smash{{\SetFigFont{20}{24.0}{\familydefault}{\mddefault}{\updefault}{\color[rgb]{0,0,0}$e_{x_\ell}^1$}%
}}}}
\put(8251,-5536){\makebox(0,0)[lb]{\smash{{\SetFigFont{20}{24.0}{\familydefault}{\mddefault}{\updefault}{\color[rgb]{0,0,0}$e_0$}%
}}}}
\put(8251,-2836){\makebox(0,0)[lb]{\smash{{\SetFigFont{20}{24.0}{\familydefault}{\mddefault}{\updefault}{\color[rgb]{0,0,0}$e_4$}%
}}}}
\put(6451,-2836){\makebox(0,0)[lb]{\smash{{\SetFigFont{20}{24.0}{\familydefault}{\mddefault}{\updefault}{\color[rgb]{0,0,0}$u_2$}%
}}}}
\put(6451,-5536){\makebox(0,0)[lb]{\smash{{\SetFigFont{20}{24.0}{\familydefault}{\mddefault}{\updefault}{\color[rgb]{0,0,0}$u_0$}%
}}}}
\put(15131,-1036){\makebox(0,0)[lb]{\smash{{\SetFigFont{20}{24.0}{\familydefault}{\mddefault}{\updefault}{\color[rgb]{0,0,0}$u_{c_1}$}%
}}}}
\put(15131,-2536){\makebox(0,0)[lb]{\smash{{\SetFigFont{20}{24.0}{\familydefault}{\mddefault}{\updefault}{\color[rgb]{0,0,0}$u_{c_2}$}%
}}}}
\put(15131,-5236){\makebox(0,0)[lb]{\smash{{\SetFigFont{20}{24.0}{\familydefault}{\mddefault}{\updefault}{\color[rgb]{0,0,0}$u_{c_k}$}%
}}}}
\put(8401,-2386){\makebox(0,0)[b]{\smash{{\SetFigFont{14}{16.8}{\familydefault}{\mddefault}{\updefault}{\color[rgb]{0,0,0}$1$}%
}}}}
\put(12301,314){\makebox(0,0)[b]{\smash{{\SetFigFont{14}{16.8}{\familydefault}{\mddefault}{\updefault}{\color[rgb]{0,0,0}$\beta$}%
}}}}
\put(8401,-5086){\makebox(0,0)[b]{\smash{{\SetFigFont{14}{16.8}{\familydefault}{\mddefault}{\updefault}{\color[rgb]{0,0,0}$M$}%
}}}}
\put(6451,-1636){\makebox(0,0)[lb]{\smash{{\SetFigFont{20}{24.0}{\familydefault}{\mddefault}{\updefault}{\color[rgb]{0,0,0}$u_1$}%
}}}}
\put(4651,-2236){\makebox(0,0)[lb]{\smash{{\SetFigFont{20}{24.0}{\familydefault}{\mddefault}{\updefault}{\color[rgb]{0,0,0}$e_2$}%
}}}}
\put(4651,-1036){\makebox(0,0)[lb]{\smash{{\SetFigFont{20}{24.0}{\familydefault}{\mddefault}{\updefault}{\color[rgb]{0,0,0}$e_1$}%
}}}}
\put(4801,-586){\makebox(0,0)[b]{\smash{{\SetFigFont{14}{16.8}{\familydefault}{\mddefault}{\updefault}{\color[rgb]{0,0,0}$1$}%
}}}}
\put(4801,-1786){\makebox(0,0)[b]{\smash{{\SetFigFont{14}{16.8}{\familydefault}{\mddefault}{\updefault}{\color[rgb]{0,0,0}$1$}%
}}}}
\put(4801,-2986){\makebox(0,0)[b]{\smash{{\SetFigFont{14}{16.8}{\familydefault}{\mddefault}{\updefault}{\color[rgb]{0,0,0}$1$}%
}}}}
\put(4651,-3436){\makebox(0,0)[lb]{\smash{{\SetFigFont{20}{24.0}{\familydefault}{\mddefault}{\updefault}{\color[rgb]{0,0,0}$e_3$}%
}}}}
\end{picture}%
}%%%%%%%%%%%%%%%%%%%%%%%%%%%%%%%%%%%%%%%%%% End resizebox
\makeatother
\end{center}
\caption{Construction for the proof of Theorem~\ref{t:npcscg}}
\label{f:construct}
\end{figure*}

\noindent
\underline{Part (a)}:~
We will construct a malicious Bayesian congestion game with singleton strategy
sets and identical type probability $p$.
Our construction imposes 
one player $u_c$ for each clause $c\in C$, 
one player $u_x$ and two resources $e_x^0,e_x^1$ for each variable $x \in X$, 
3 additional players $u_0,u_1,u_2$, 
and 5 additional resources $e_0,e_1,e_2,e_3,e_4$.
Our construction is summarized in Figure~\ref{f:construct}. Resources
are depicted as squares and players as circles and an edge (solid or dotted) 
between a resource $e$ and a player $u$ indicates that $\{e\}$ is 
in $u$'s strategy set. 
A number $\alpha$ above a resource $e$ defines 
the slope of the corresponding linear
latency function $f_e(\delta)= \alpha \cdot \delta$.
Denote $ E_v=\{e_{x_1}^0, e_{x_1}^1, \ldots , e_{x_\ell}^0, e_{x_\ell}^1\}$.
For the proof of part (a), let $\beta=2-p$. So,
all resources $e\in E_v$
share the latency function $f_e(\delta)=(2-p)\cdot \delta$.

Player $u_0$ can only be assigned to $e_0$. Both $u_0$ and $e_0$ are used
to collect the malicious type-agents of all players except player $u_1$. Thus all
those players have $e_0$ in their strategy set and $M$ is chosen sufficiently large, such that 
for all those malicious type-agents $e_0$ is a dominant strategy and no selfish type other than
$u_0^s$ will ever prefer to choose $e_0$. Choosing $M=\ell+1$ suffices. Player $u_1$ and $u_2$
are connected to $e_1$, $e_2$, and $e_3$, while $u_2$ can also choose $e_0$ and $e_4$.
For each variable $x\in X$, the corresponding {\em variable  player}
$u_x$ is connected to $e_0$, $e_4$, $e_x^0$ and $e_x^1$. Assigning the selfish
type-agent  
$u^s_x$ to  $e_x^0$ (resp. $e_x^1$) will be interpreted as setting $x$ to {\sf true}
(resp. {\sf false}).
For each clause $c\in C$, the corresponding {\em clause player}
 $u_c$ is connected
to $e_0$ and to all resources $e_x^0$ ($e_x^1$) with $x\in X$ and $x$ appears
{\em negated} ({\em unnegated}) in $c$. For the example in 
Figure~\ref{f:construct}, 
$c_1=(\xn_1 \vee x_2 \vee \xn_\ell)$, 
$c_2=(\xn_1 \vee \xn_2)$, and
$c_k=(x_1   \vee x_2 \vee x_\ell)$.
Observe that by the structure of our 3-SAT instance, no more than 
two clause players are connected to each resource in $E_v$.
This finishes the construction of the malicious Bayesian congestion game.

We will first show that if the 3-SAT instance is satisfiable then the corresponding Bayesian congestion
game possesses a pure Bayesian Nash equilibrium.
Given a satisfying truth assignment, we define a strategy profile 
$\vsi$ of the malicious Bayesian congestion game as follows: 
\begin{itemize}
\item
Both type-agents of player $u_0$ can only be assigned to $e_0$. 
\item
All malicious type-agents except $u^m_1$ are assigned to
resource $e_0$. By the choice of $M$, none of those malicious type-agents can improve.  
\item
Both type-agents of player $u_1$ are assigned to $e_1$ and no type-agent of any player 
is assigned to $e_2$ or $e_3$.
It is easy to see that neither $u^m_1$ nor $u^s_1$ have an incentive to switch.
\item 
Type agent $u^s_2$ is the only type-agent assigned to $e_4$. 
So,  $u^s_2$ cannot improve.
\item 
For each $x\in X$, the selfish type-agent $u^s_x$ of variable player $u_x$ is assigned to 
resource $e_x^0$ if $x={\sf true}$ in the satisfying truth assignment, and to $e_x^1$ otherwise. 
Each of these selfish type-agents is the only type-agent assigned to her resource. 
So, they all experience an 
expected  latency of $\beta=2-p$ and changing to $e_4$ would yield 
the same expected latency.
Thus, the selfish type-agents of all variable players are satisfied.
\item 
Denote by $E'_v$ the subset of resources from $E_v$ to which no selfish type-agent of a  variable player
is assigned. 
Since we have a satisfying truth assignment, each clause player is connected
to some resource from $E'_v$. 
For each $c \in C$ , the selfish type-agent $u^s_c$ is assigned to some resource in $E'_v$ as follows:
\\
Consider the sub-game that consists only of the selfish type-agents 
of the clause players $u_c$, $c\in C$ and
the set of resources $E'_v$. Observe that this sub-game is a (non-malicious) 
congestion game and thus admits 
a pure Nash equilibrium \cite{Ros73a}. 
Assign the selfish type-agents of each clause player according to this Nash
equilibrium. So, none of these selfish type-agents can improve by changing to some other resource in 
$E'_v$. Moreover, at most two selfish type-agents are assigned to each resource in $E'_v$ and 
there is exactly 
one selfish type-agent of a variable player assigned to each 
resource in $E_v\setminus E'_v$. Thus, the
selfish type-agents of all clause players are satisfied.
\end{itemize}
Since no type-agent can improve be changing her strategy, it follows that $\vsi$ in a pure Bayesian Nash
equilibrium.

For the other direction observe that any pure Bayesian Nash equilibrium $\vsi$ fulfills the following structural properties:
\begin{enumerate}
\item[(I)]
All malicious type-agents except $u^m_1$ are assigned to
resource $e_0$ and $u^s_0$ is the only selfish type-agent assigned to $e_0$. 
\item[(II)]
The selfish type-agent $u^s_2$ is assigned to $e_4$ and no other type-agent is assigned to $e_4$.
\end{enumerate}
Property (I) follows immediately by the choice of $M$. We will now prove property (II).

By way of contradiction assume that $u^s_2$ is assigned to a resource in 
$\{e_1,e_2,e_3\}$ in a pure Bayesian Nash equilibrium $\vsi$. In this case $u^m_1$
will always choose the same resource  as $u^s_2$. 
However, then there must be an empty resource in $\{e_1,e_2,e_3\}$ and $u^s_2$
can improve by choosing this empty resource.
This contradicts our assumption that $\vsi$ is a pure Bayesian Nash 
equilibrium. Thus, 
$u^s_2$ is assigned to $e_4$. If some other type-agent is also assigned 
to $e_4$, then $u^s_2$ experiences an expected latency of at least $2-p$
and $u^s_2$ could decrease her expected latency to $1$ by switching to the 
empty resource in $\{e_1,e_2,e_3\}$. Again a contradiction to $\vsi$ being
a pure Bayesian Nash equilibrium. 
It follows that $u^s_2$ is the only type-agent assigned to $e_4$ in $\vsi$.
This completes the proof of property (II).

Since $u^s_2$ is the only type-agent assigned to $e_4$ it follows that for each variable $x\in X$ the 
corresponding selfish type-agent $u^s_x$ is either assigned to $e_x^0$ or to $e_x^1$.
If $u^s_x$ is not the only type-agent on that resource then her expected latency is at least $(2-p)^2$
while changing to $e_4$ would improve her expected latency to $2-p$,  
a contradiction to $\vsi$  being a pure Bayesian Nash equilibrium. 
It follows that 
the selfish type-agents of all clause players are only assigned to resources in $E_v$ to which no 
selfish type-agent of a variable player is assigned.  This is only possible if the strategies of 
the selfish type-agents $u^s_x, x\in X$ correspond to a satisfying truth assignment.
This finishes the proof of part (a).

\noindent
\underline{Part (b)}:~
To see that (b) holds we alter the construction depicted in 
Figure~\ref{f:construct} slightly
by deleting player $u_0$ and resource $e_0$. 
Furthermore, in the new construction player $u_1$ is the only 
player that is malicious with positive probability $p$. 
For the slope of the latency functions of resources in $E_v$, let $\beta=\frac{3}{2}$ 
(in fact any $1<\beta<2$ would also do).
The rest of the construction does not change. The proof now follows
the same line of arguments as in part (a) with only minor changes.
\qed
\end{proof}

\begin{theorem}
\label{t:npcscg_ex}
The results from Theorem~\ref{t:npcscg} hold, 
even if $|S_u|\le 4$ for all players $u\in \N$ and for each resource $e\in E$ there are at most 
three players $u\in\N$ with $\{e\} \in S_u$.
\end{theorem}

\begin{proof} [Sketch]
We will slightly alter the construction from Figure~\ref{f:construct}. First observe that we have 
already $|S_u|\le 4$ for all players $u\in \N$. Furthermore, the only resources that are in the strategy
set of more than three players are $e_3$ and for part (a) also $e_0$.

\begin{figure}[ht]
\begin{center}
\makeatletter
\resizebox*{0.4\textwidth}{!}{%
\begin{picture}(0,0)%
\includegraphics{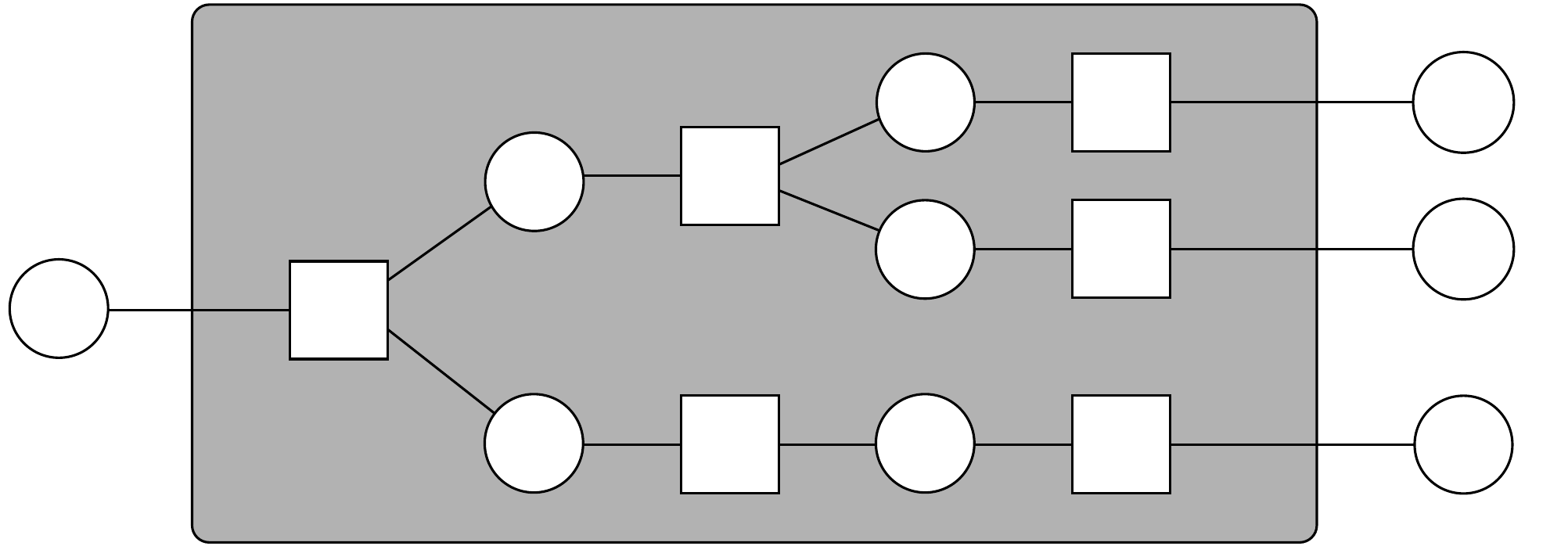}%
\end{picture}%
\setlength{\unitlength}{3947sp}%
\begingroup\makeatletter\ifx\SetFigFont\undefined%
\gdef\SetFigFont#1#2#3#4#5{%
  \reset@font\fontsize{#1}{#2pt}%
  \fontfamily{#3}\fontseries{#4}\fontshape{#5}%
  \selectfont}%
\fi\endgroup%
\begin{picture}(9609,3344)(622,-3683)
\put(2666,-2311){\makebox(0,0)[b]{\smash{{\SetFigFont{20}{24.0}{\familydefault}{\mddefault}{\updefault}{\color[rgb]{0,0,0}$e_{3,0}$}%
}}}}
\put(7511,-1036){\makebox(0,0)[b]{\smash{{\SetFigFont{20}{24.0}{\familydefault}{\mddefault}{\updefault}{\color[rgb]{0,0,0}$e_{3,x_1}$}%
}}}}
\put(7511,-1936){\makebox(0,0)[b]{\smash{{\SetFigFont{20}{24.0}{\familydefault}{\mddefault}{\updefault}{\color[rgb]{0,0,0}$e_{3,x_2}$}%
}}}}
\put(7511,-3136){\makebox(0,0)[b]{\smash{{\SetFigFont{20}{24.0}{\familydefault}{\mddefault}{\updefault}{\color[rgb]{0,0,0}$e_{3,x_3}$}%
}}}}
\put(2701,-1861){\makebox(0,0)[b]{\smash{{\SetFigFont{14}{16.8}{\familydefault}{\mddefault}{\updefault}{\color[rgb]{0,0,0}$1$}%
}}}}
\put(5126,-1036){\makebox(0,0)[b]{\smash{{\SetFigFont{14}{16.8}{\familydefault}{\mddefault}{\updefault}{\color[rgb]{0,0,0}$\beta$}%
}}}}
\put(7576,-586){\makebox(0,0)[b]{\smash{{\SetFigFont{14}{16.8}{\familydefault}{\mddefault}{\updefault}{\color[rgb]{0,0,0}$\beta^2$}%
}}}}
\put(976,-2311){\makebox(0,0)[b]{\smash{{\SetFigFont{20}{24.0}{\familydefault}{\mddefault}{\updefault}{\color[rgb]{0,0,0}$u_2$}%
}}}}
\put(9601,-1036){\makebox(0,0)[b]{\smash{{\SetFigFont{20}{24.0}{\familydefault}{\mddefault}{\updefault}{\color[rgb]{0,0,0}$u_{x_1}$}%
}}}}
\put(9601,-1936){\makebox(0,0)[b]{\smash{{\SetFigFont{20}{24.0}{\familydefault}{\mddefault}{\updefault}{\color[rgb]{0,0,0}$u_{x_2}$}%
}}}}
\put(9601,-3136){\makebox(0,0)[b]{\smash{{\SetFigFont{20}{24.0}{\familydefault}{\mddefault}{\updefault}{\color[rgb]{0,0,0}$u_{x_3}$}%
}}}}
\end{picture}%
}%%%%%%%%%%%%%%%%%%%%%%%%%%%%%%%%%%%%%%%%%% End resizebox
\makeatother
\end{center}
\caption{Tree for $\ell=3$}
\label{f:tree}
\end{figure}

To resolve this for $e_3$, disconnect all players
from $e_3$ and replace the single resource $e_3$ with a binary tree of resources with root $e_{3,0}$ that has $\ell$ leaves $e_{3,x_1}, \ldots e_{3,x_\ell}$, all with depth $\lceil \log(\ell)\rceil$. For a resource $e$ 
at level $j$ the latency function is defined by $f_e(\delta)=\beta^{j} \cdot \delta$. 
So $f_{e_{3,0}}(\delta)= 1$  and $f_{e_{3,x}}(\delta)= \beta^{\lceil\log(\ell)\rceil} \cdot \delta$ 
for all leaves $x\in X$.
For each pair of resources from two consecutive levels, we introduce a new player to connect them. Call those players \emph{tree players}. Figure~\ref{f:tree} shows the construction for $\ell =3$.
Player $u_2$ is connected to resource $e_{3,0}$ and each variable player $x\in X$ is connected to
$e_{3,x}$. We also change the latency function of all resources $e\in E_v$ (cf. Theorem~\ref{t:npcscg}) 
to $f_e(\delta)=\beta^{\lceil\log(\ell)\rceil}\cdot \delta$.

Moreover, for part (a) we have to resolve that more than three players are connected to $e_0$.
To do so, we simply copy resource $e_0$
together with player $u_0$ multiple times and connect all players (including the tree players) except 
player $u_1$ to the new set of resources that evolve from $e_0$. By having sufficiently many copies, this
can be done, such that no more than three players are connected to each new resource. Again, $M$ is 
chosen sufficiently large, e.g. $M=2^{\lceil\log(\ell)\rceil +1}$.   

Observe that $u_2^s$ will only selfishly choose $e_{3,0}$ if all
tree players choose the strategy that is closer to the leaves. 
The rest of the proof now simply follows the proof of Theorem~\ref{t:npcscg}.
\qed
\end{proof}

For the more restricted class of symmetric malicious Bayesian congestion game with singleton 
strategy sets, identical type probability $p$ and identical latency functions we can easily decide
whether a pure Bayesian Nash equilibrium exists or not.

\begin{theorem}
\label{t:characterization}
A symmetric malicious Bayesian congestion game with singleton strategy sets, identical 
type probability $p$ and identical (not necessarily linear) latency 
functions possesses a pure Bayesian Nash equilibrium if and only if either of the following 
properties holds:
\begin{enumerate}
\item[(a)]
$p\le \frac{1}{2}$ and $r=2$
\item[(b)]
$p\le \frac{1}{2}$ and $r | n$
\end{enumerate}
\end{theorem}

\begin{proof}
Assume that at least one of the properties holds. We will show that in each case this implies the
existence of a pure Bayesian Nash equilibrium. 

First assume that (a) holds:  
For each player $u_i$ assign the selfish type-agent $u_i^s$ alternately to the two resources. At each time, we assign the corresponding malicious type-agent to the other resource. 
It's not hard to see that the resulting strategy profile is a pure Bayesian Nash equilibrium.

Now assume that (b) holds:
In this case assign $\frac{n}{r}$ selfish type-agents and $\frac{n}{r}$ malicious type-agents to each resource such that the selfish and malicious type-agent of each fixed player are not assigned to the same
resource. Again it is easy to see that the resulting strategy profile is a pure Bayesian Nash equilibrium.

For the other direction we will show that if neither (a) nor (b) holds then the malicious Bayesian congestion game does not possess a pure Bayesian Nash equilibrium. By way of contradiction assume there exists 
a  malicious Bayesian congestion game $\Gamma$ satisfying neither (a) nor (b) but $\Gamma$ admits 
a pure Bayesian Nash equilibrium $\vsi$. 
We consider 3 sub-cases:
\medskip

\noindent
\underline{Case 1:} $p>\frac{1}{2}$ and $r=2$\\
By way of contradiction assume $\vsi$ assigns more than  $\lceil\frac{n}{2}\rceil$ selfish type-agents
to some resource $e$. If this is the case then $\vsi$ will also assign all malicious type-agents to $e$.
But then all the selfish type-agents on resource $e$ can improve by switching to the other resource, 
a contradiction to $\vsi$ being a pure Bayesian Nash equilibrium.
It follows that at most $\lceil\frac{n}{2}\rceil$ selfish type-agents are assigned to each resource. This
again implies that  $\lceil\frac{n}{2}\rceil$ selfish type-agents are assigned to one resource (say $e_1$)
and $\lfloor \frac{n}{2} \rfloor$ selfish type-agents are assigned to the other resource (say $e_2$).
Denote $\N_1=\{u\in \N | \sigma(u^s) = e_1\}$ the set of players with selfish type-agent assigned to
resource $e_1$ and denote $\N_2=\N\setminus \N_1$.
For each player $u \in \N_2$ we have $\lsel_{e_1}(\vsi_{-u})> \lsel_{e_2}(\vsi_{-u})$ which implies
that $\vsi(u^m)=e_1$ for all $u\in \N_2$. 
Will now show that $\vsi(u^m)=e_2$ for all $u\in \N_1$. If $n$ is even then this holds immediately by
symmetry. So assume $n$ is odd and $\exists u\in N_1$ with $\vsi(u^m)=e_1$.
Now consider an arbitrary player $u'\in \N_1$ with $u'\neq u$. Since $n$ is odd it follows that $n\ge 3$ and thus such a player exits. Furthermore, it follows that $\ntf=\ntc -1$. Then
\begin{align*}
   \lsel_{e_1}(\vsi_{-u'})+\lmal_{e_1}(\vsi_{-u'})
     &\ge (1-p)\cdot (\ntc -1)
              + p \cdot (\ntf + 1)\\
%      &= \ntc -1
%        + p \cdot 
%         \left(
%            - \ntc + \ntf +2
%         \right)\\
       &= \ntc -1 + p
\end{align*}
while
\begin{align*}
   \lsel_{e_2}(\vsi_{-u'})+\lmal_{e_2}(\vsi_{-u'})
      &\le (1-p)\cdot \ntf
               + p \cdot (\ntc -2)\\
%      & = \ntf + p \cdot 
%            \left(
%                \ntc - \ntf -2 
%             \right)\\
%       & = \ntf - p \\
       & = \ntc -1 - p\\
       & < \lsel_{e_1}(\vsi_{-u'})+\lmal_{e_1}(\vsi_{-u'}),
\end{align*}
a contradiction to $\vsi$ being a pure Bayesian Nash equilibrium.
So, $\vsi(u^m)=e_2$ for all $u\in \N_1$. 
Summing up, for all $u\in N_1$ we have $\sigma(u^s)=e_1$ and $\sigma(u^m)=e_2$ while
for all $u\in N_2$ we have $\sigma(u^s)=e_2$ and $\sigma(u^m)=e_1$. 
In such an assignment each selfish type-agent on resource $e_2$ can improve be switching to
$e_1$. This contradicts our initial assumption that 
$\vsi$ is a pure Bayesian Nash equilibrium.
\medskip

\noindent
\underline{Case 2:} $p>\frac{1}{2}$ and $r | n$\\
First assume by way of contradiction that there exists a resource $e\in E$ to which $\vsi$ assigns 
more than $\frac{n}{r}$ selfish type-agents. It follows that there also exists some other resource $e'$ 
to which $\vsi$ assigns less than $\frac{n}{r}$ selfish type-agents. Since $\vsi$ is a pure Bayesian
Nash equilibrium, no malicious type-agent is assigned to $e'$. However, then all selfish type-agents on $e$ improve
by switching to $e'$, a contradiction. It follows that $\vsi$ assigns 
exactly $\frac{n}{r}$ selfish type-agents to each resource. This also implies that 
$\sigma(u^m)\neq \sigma(u^s)$ for all players $u\in\N$.

If more than $\frac{n}{r}$ malicious type-agents are assigned to some resource $e\in E$ then all selfish
type-agents on $e$ can improve since $p>\frac{1}{2}$. So, $\vsi$ assigns 
exactly $\frac{n}{r}$ malicious type-agents to each resource. 
However, in such a pure strategy profile the selfish type-agent $u^s$ of each player $u$ can improve by switching to $\sigma(u^m)$. This contradicts our initial assumption that 
$\vsi$ is a pure Bayesian Nash equilibrium.
\medskip

\noindent
\underline{Case 3:} $r\ge 3$ and $\frac{n}{r}\not\in \mathbb{N}$\\
First observe that if $n<r$ and the malicious type-agents are satisfied then there is always some resource 
$e\in E$ to which no type-agent is assigned. Each selfish type-agent can then improve by switching to
$e$, a contradiction to $\vsi$ being a pure Bayesian Nash equilibrium.
So we may assume that $n>r$.

Let $E^+$ be the set of resources to which $\vsi$ assigns at least \nmc
selfish type-agents. Since $\frac{n}{r}\not\in \mathbb{N}$ it follows that
$1\le |E^+|\le r-1$. If $|E^+|\ge 2$ then $\vsi$ assigns all malicious type-agents to a resource in $E^+$.
This implies that there exists a selfish type-agents assigned to some resource in $E^+$ that can improve
by switching to some resource in $E\setminus E^+$. It follows that $|E^+|=1$. Without loss of generality 
assume $E^+=\{e_1\}$. 

Since $|E^+|=1$ it follows that $\vsi$ assigns exactly $\nmf=\nmc-1$ selfish type-agents to each resource
$e\in E\setminus E^+$.
Now, for all $u\in N$, if $\sigma(u^s)\in E\setminus E^+$ then $\sigma(u^m)=e_1$. It follows that $\vsi$
assigns at least $(r-1)\cdot \nmf$ malicious type-agents to $e_1$ and (by the pigeon hole principle) 
there exists a resource $e'\in E\setminus E^+$ to which $\vsi$ assigns at most 
$\left\lfloor\frac{\nmc}{r-1}\right\rfloor$ malicious type-agents.  Since $r\ge 3$ and $n> r$ it follows 
that all selfish type-agents on resource $e_1$ can improve by switching to resource $e'$.
This contradicts our initial assumption that $\vsi$ is a pure Bayesian Nash equilibrium.

In each case we got a contradiction to our assumption that $\vsi$ is a pure Bayesian Nash equilibrium, 
proving that $\Gamma$ does not admit a pure Bayesian Nash equilibrium. This finishes the proof of the 
theorem.
\qed
\end{proof}
Observe, that the previous proof is constructive. So, if the requirements for the existence of a
pure Bayesian Nash equilibrium are fulfilled, then this equilibrium can also be easily constructed in 
linear time.

\section{Price of Malice}
\label{s:poa}
We now shift gears and present our results that are related to the Price of Malice.
We start with a general upper bound on the Price of Byzantine Anarchy. 
The proof of this upper bound uses a technique from \cite{CK05} adapted to the model of 
malicious Bayesian congestion games. 
Furthermore, it makes use of the following technical lemma, which has an easy proof.

\begin{lemma}
\label{l:ineq}
For all $x,y\in\mathbb{R}$ and $c>0$ we have 
$x\cdot y\le c\cdot x^2 + \frac{1}{4c} \cdot y^2$
\end{lemma}

%We are ready to prove:

\begin{theorem}
\label{t:pob}
Consider the class of malicious Bayesian congestion games $\games(\tsl)$ with affine latency 
functions.
Then,
$$
    \pob(\tsl) \le \frac{n}{n-\tsl} (1-p_{\min}) \left(\tsl + \frac{3+\sqrt{5+4\tsl}}{2} \right).
$$
\end{theorem}
\begin{proof}
Let $\Psi$ be an arbitrary malicious Bayesian congestion game from $\games(\tsl)$
and let $\Gamma_\Psi$ be the corresponding (non-malicious)
congestion game.
Let $\Q$ be an arbitrary Bayesian Nash equilibrium for $\Psi$. 
Furthermore, let $\s^*$ be an optimum pure strategy profile for $\Gamma_\Psi$.
For each player $u\in\N$, we have
\begin{align*}
\PC_u(\p,\Q) 
&=   \sum_{\vsi \in\S^2} q(\vsi) \cdot \PC_u(\p,\vsi)\\
&\le \PC_u(\p,(\Q_{-u^s},s^*_u))\\
&= \sum_{\vsi \in\S^2} q(\vsi)
   \cdot \PC_u(\p,(\vsi_{-u^s},s^*_u))\\
&= \sum_{\vsi \in\S^2} q(\vsi)
   \sum_{e\in s^*_u} f_e( \lsel_e(\vsi_{-u}) + \lmal_e(\vsi_{-u})+1)\\
&\le  \sum_{\vsi \in\S^2} q(\vsi)
   \sum_{e\in s^*_u} f_e( \lsel_e(\vsi) + \tsl +1),
\end{align*}
where the first inequality follows since $\Q$ is a Bayesian Nash equilibrium and the 
second inequality holds,  
since $\lmal_e(\vsi_{-u})\le \lmal_e(\vsi)  \le \tsl$
for all $e\in E$. 
So, we get
\begin{align*}
%\lefteqn
{\sum_{u\in\N} (1-p_u) \cdot \PC_u(\p,\Q)}% \\%\hspace{2em} \\
&\le \sum_{\vsi \in\S^2} q(\vsi)
     \sum_{u\in\N} \sum_{e\in \sigma^*(u)} 
             (1-p_u) \cdot f_e( \lsel_e(\vsi) + \tsl + 1)\\
&\le \sum_{\vsi \in\S^2} q(\vsi)
   \sum_{e\in E} (1-p_{\min}) \cdot \lsel_e(\s^*) \cdot f_e( \lsel_e(\vsi) + \tsl + 1)\\
&\le (1-p_{\min}) \cdot \sum_{\vsi \in\S^2} q(\vsi)
   \sum_{e\in E} \lsel_e(\s^*) \cdot f_e(\tsl+1)\\
&\phantom{\le} + (1-p_{\min}) \cdot \sum_{\vsi \in\S^2} q(\vsi)
\sum_{e\in E} a_e \cdot \lsel_e(\vsi) \cdot \lsel_e(\s^*)\\
\end{align*}
Observe, that $\lsel_e(\s^*)\ge 1$ and thus 
$\lsel_e(\s^*)\le \lsel_e(\s^*)^2$. Moreover, by applying
Lemma~\ref{l:ineq} with $x=\lsel_e(\vsi)$ and $y=\lsel_e(\s^*)$, we get 
\begin{align*}
\lefteqn{\sum_{u\in\N} (1-p_u) \cdot \PC_u(\p,\Q)} \\%\hspace{2em} \\
&\le(1-p_{\min}) \cdot (\tsl+1)\cdot \sum_{\vsi \in\S^2} q(\vsi)
   \sum_{e\in E} \lsel_e(\s^*) \cdot f_e( \lsel_e(\s^*))\\
&\phantom{=} + (1-p_{\min}) \cdot c\cdot \sum_{\vsi \in\S^2} q(\vsi)
         \sum_{e\in E} a_e \cdot \lsel_e(\vsi)^2\\
&\phantom{=} + (1-p_{\min}) \cdot \frac{1}{4c}\cdot \sum_{\vsi \in\S^2} q(\vsi)
         \sum_{e\in E} a_e \cdot \lsel_e(\s^*)^2\\
&\le (1-p_{\min}) (\tsl+1+\frac{1}{4c}) \sum_{\vsi \in\S^2} q(\vsi)
   \sum_{e\in E} \lsel_e(\s^*) \cdot f_e( \lsel_e(\s^*))\\
&\phantom{=} + c\cdot \sum_{\vsi \in\S^2} q(\vsi)
         \sum_{e\in E}  \lsel_e(\vsi) f_e(\lsel_e(\vsi))\\
&= (1-p_{\min}) (\tsl+1+\frac{1}{4c})\cdot 
   \sum_{u\in\N} \sum_{e\in s^*_u}  f_e( \lsel_e(\s^*))\\
&\phantom{=} + c\cdot \sum_{\vsi \in\S^2} q(\vsi)
         \sum_{u\in\N} \sum_{e\in \sigma(u^s)} (1-p_u) \cdot f_e(\lsel_e(\vsi))\\
&\le (1-p_{\min}) (\tsl+1+\frac{1}{4c})\cdot 
   \sum_{u\in\N} \sum_{e\in s^*_u}  f_e( \lsel_e(\s^*))\\
&\phantom{=} + c\cdot \sum_{\vsi \in\S^2} q(\vsi)
         \sum_{u\in\N} \sum_{e\in \sigma(u^s)} (1-p_u) \cdot f_e(\lsel_e(\vsi_{-u})+1)\\
&\le (1-p_{\min})(\tsl+1+\frac{1}{4c})\cdot  \sum_{u\in\N} \PC_u(\s^*) 
\\&\phantom{\le} 
 + c\cdot   \sum_{u\in\N} (1-p_u) \cdot \PC_u(\p,\Q) \label{e:1} \\
\end{align*}
It follows that 
\begin{align}
%\lefteqn
   {\frac{\SC(\p,\Q)}
        {\SC(\s^*)}}%\\
&=
\frac{n}{n-\tsl} \cdot \frac{\sum_{u\in\N} (1-p_u) \cdot \PC_u(\Q)}{\sum_{u\in\N} \PC_u(\s^*)} \nonumber\\
&\le \frac{n}{n-\tsl} (1-p_{\min})\frac{\tsl+1+\frac{1}{4c}}
         {1-c}
\end{align}
Now, choosing $c=\frac{-1+\sqrt{5+4\tsl}}{4(\tsl+1)}$ yields 
\begin{align}
\frac{\tsl+1+\frac{1}{4c}}
         {1-c}
&= \frac{\tsl+1+\frac{\tsl+1}{-1+\sqrt{5+4\tsl}}}
         {1-\frac{-1+\sqrt{5+4\tsl}}{4(\tsl+1)}} 
         \nonumber \\
&=\frac{4(\tsl+1)^2(1+\frac{1}{-1+\sqrt{5+4\tsl}})}
            {5+4\tsl-\sqrt{5+4\tsl}}
            \nonumber \\
&=\frac{4(\tsl+1)^2(1+\frac{1+\sqrt{5+4\tsl}}{4+4\tsl})}
            {5+4\tsl-\sqrt{5+4\tsl}}
            \nonumber \\
 &=\frac{(\tsl+1)(5+4\tsl+\sqrt{5+4\tsl})}
            {5+4\tsl-\sqrt{5+4\tsl}}
            \nonumber \\
 &=\frac{(\tsl+1)(\sqrt{5+4\tsl}+1)}
            {\sqrt{5+4\tsl}-1}
            \nonumber \\
 &=\frac{(\tsl+1)(\sqrt{5+4\tsl}+1)^2}
            {4+4\tsl}
            \nonumber \\
  &=\frac{1}{4} (4\tsl+ 6 + 2\sqrt{5+4\tsl})
  \nonumber \\
  &= \tsl + \frac{3+ \sqrt{5+4\tsl}}{2} \label{e:2}
\end{align}
The theorem follows by combining (\ref{e:1}) and (\ref{e:2}) and since $\Q$ is an arbitrary Bayesian
Nash equilibrium. 
%Or more precisely
%\begin{align*}
%\frac{\tsl+\y+\frac{1}{4c}}
%         {1-c}
%&= \frac{\tsl+\y+\frac{\tsl+\y}{-1+\sqrt{4+\y+4\tsl}}}
%         {1-\frac{-1+\sqrt{4+\y+4\tsl}}{4(\tsl+\y)}}\\
%&=\frac{4(\tsl+\y)^2(1+\frac{1}{-1+\sqrt{4+\y+4\tsl}})}
%            {4+\y+4\tsl-\sqrt{4+\y+4\tsl}}\\
%&=\frac{4(\tsl+\y)^2(1+\frac{1+\sqrt{4+\y+4\tsl}}{3+\y+4\tsl})}
%            {4+\y+4\tsl-\sqrt{4+\y+4\tsl}}\\
% &=\frac{(\tsl+\y)(4+\y+4\tsl+\sqrt{4+\y+4\tsl})}
%            {4+\y+4\tsl-\sqrt{4+\y+4\tsl}}\\
% &=\frac{(\tsl+1)(\sqrt{5+4\tsl}+1)}
%            {\sqrt{5+4\tsl}-1}\\
% &=\frac{(\tsl+1)(\sqrt{5+4\tsl}+1)^2}
%            {4+4\tsl}\\
%  &=\frac{1}{4} (4\tsl+ 6 + 2\sqrt{5+4\tsl})\\
%  &= \tsl + \frac{3+ \sqrt{5+4\tsl}}{2} 
%\end{align*}
\qed
\end{proof}

For the case of identical type probabilities we can provide a better upper bound on the Price of Byzantine Anarchy. 
Observe that for identical type probabilities, $\tsl=p\cdot n$ and $p_{\min}=p$.
As an immediate corollary to Theorem~\ref{t:pob}, we get:

\begin{corollary}
\label{c:pob}
Consider the class of malicious Bayesian congestion games $\games(\tsl)$ 
with affine latency functions and identical type probability $p$.
Then,
$$
    \pob(\tsl) \le \tsl + \frac{3+\sqrt{5+4\tsl}}{2}.
$$
\end{corollary}

We proceed by introducing a malicious Bayesian congestion game that is parameterized by 
a parameter $\alpha$. In the remainder of the paper, we will make use of this construction twice, 
each time with a different parameter $\alpha$.

\begin{example}
\label{ex}
Given some $\alpha>0$,
construct a
malicious Bayesian congestion game $\Gamma(\alpha)$ with linear latency functions, 
$n\ge 3$ players and identical
type probability $p$ and $|E|=2n$ as follows:
Let $E=E_1\cup E_2$ with 
$E_1 = \{ g_1, \dots, g_n \}$ and
$E_2 = \{ h_1, \dots, h_n \}$.
Each player $u\in\{1,\ldots n\}$ has three strategies in her strategy set. 
So, $S_u=\{s_u^1, s_u^2, s_u^3\}$ with
$s_u^1=\{ g_u, h_u \}, s_u^2 = \{g_{u+1},h_{u+1},h_{u+2}\}$ and 
$s_u^3=E_1\cup E_2$,
where
$g_j = g_{j - n}$ and $h_j = h_{j - n}$
for $j > n$.
\\
Each resource $e\in E_1$  has a latency function $f_e(\delta)= \alpha\cdot\delta$ 
whereas the resources $e\in E_2$ share the identity as their latency
function, i.e. $f_e(\delta)=\delta$.
\end{example} 

The following theorem makes use of Example~\ref{ex} to show the following lower bound 
on the Price of Byzantine Anarchy.

\begin{theorem}
\label{t:lbi}
Consider the class of malicious Bayesian congestion games $\games(\tsl)$ with
linear latency functions and
identical type probability $p$.
Then ,
 $$
    \pob(\tsl) \ge \tsl + 2.
$$
\end{theorem}
\begin{proof}
Consider the malicious Bayesian congestion game $\Psi=\Psi(\alpha)$ given in Example~\ref{ex} with 
$\alpha=\frac{1+(n-1)p}{1-p}$. 
Observe that $\Delta=n\cdot p$.

Obviously , the optimum allocation $\s^*$ for the corresponding non-malicious game $\Gamma_\Psi$ 
is for each player 
$u\in\N$ to choose strategy  $s_u^1$. This yields 
$\SC(\Gamma_\Psi,\s^*) = 1+ \alpha = \frac{2+(n-2)p}{1-p}$.  

On the other hand, if $\sigma(u^m)=s_u^3$ and $\sigma(u^s)=s_u^2$ for all player $u\in \N$, then $\vsi$
is a (pure) Bayesian Nash equilibrium for $\Psi$, with
\begin{align*}
    \SC(\Psi,\vsi) &= 2(1+(1-p)+(n-1)p) + (1+(n-1)p)\cdot \alpha \\
    &=  \frac{2(1-p)(2+(n-2)p) + (1+(n-1)p)^2}{1-p} \\
\end{align*}
It follows that
\begin{align*}
\frac{\SC(\Psi,\vsi)}{\SC(\Gamma_\Psi,\s^*)}
&= 2(1-p) + \frac{ (1+(n-1)p)^2}{2+(n-2)p}\\
&= 2(1-p) + \frac{ 1+(n-1)p (2+ (n-1)p)}{2+(n-2)p}\\
&> 2-3p +n\cdot p\\
&= \tsl + 2 -3p.
\end{align*}
The Theorem follows for $p\rightarrow 0$, which implies $n\rightarrow\infty$.
\qed
\end{proof}

Recall that the Price of Anarchy of (non-malicious) congestion games with affine latency functions is
$\frac{5}{2}$ \cite{CK05}.
By combining this with Corollary~\ref{c:pob} and Theorem~\ref{t:lbi} we get:

\begin{corollary}
\label{c:pom}
Consider the class of malicious Bayesian congestion games $\games(\tsl)$ with 
affine latency functions and
identical type probability.
Then,
$$
   \pom (\tsl) = \Theta (\tsl).
$$
\end{corollary}

For certain congestion games, introducing malicious types might also be beneficial to the system, in the
sense that the social cost of the worst case equilibrium (one that maximizes social cost) decreases. 
To capture this, we define the {\em Windfall of Malice}. The term Windfall of Malice is due to \cite{BKP07}.
For a malicious Bayesian congestion game $\Psi$, denote $\wom(\Psi)$ as the ratio between the 
worst case Nash equilibrium of the corresponding congestion game $\G_\Psi$ and 
the worst case Bayesian Nash equilibrium of $\Psi$. We show, 

\begin{theorem}
\label{t:wom}
For each $\epsilon>0$ there is a malicious Bayesian congestion game $\Psi$ with 
linear latency functions and
identical type probability, such that
$$\wom(\Psi)\ge \frac{5}{2}-\epsilon.$$
\end{theorem}
\begin{proof}
Consider the malicious Bayesian congestion game $\Psi=\Psi(\alpha)$ given in Example~\ref{ex} with 
$n=3$ and 
$\alpha=1$. This game (for $n\ge 3$) 
was already used in \cite{CK05} to proof a lower bound on the Price of Anarchy 
for the corresponding non-malicious congestion games.
For the congestion game $\Gamma_\Psi$ that corresponds to $\Psi$, 
all players $u$ choosing $s_u^2$ is a Nash equilibrium $\s$ that maximizes 
social cost and $\SC(\Gamma_\Psi,\s)=5$.

Now, consider the malicious Bayesian congestion game $\Psi$, where $p>0$.
First observe that choosing $s_u^3$ is always a strictly dominant strategy for the malicious type-agent 
$u^m$ for all $u\in\N$. Moreover, $u^s$ will never choose $s_u^3$. For $i\in\{2,3\}$, 
let $q_i$ be the probability that $u_i^s$ chooses $s_{u_i}^1$. Then $u_i^s$ chooses $s_{u_i}^2$ with 
probability $(1-q_i)$. We will show that for all $p>0$, the selfish type-agent $u_1^s$ experiences a strictly lower expected latency, if she chooses $s_{u_1}^1$ and not $s_{u_1}^2$.

On the one hand,
if $u_1^s$ chooses $s_{u_1}^1$ then her expected latency is:
\begin{align*}
&\underbrace{1+2p+(1-q_2+1-q_3)(1-p)}_{h_1} + \underbrace{1+2p + (1- q_2)(1-p)}_{g_1} \\
&= 2+ 4p +(1-p)(3 -2 q_2 - q_3)
\end{align*}
On the other hand,
if $u_1^s$ chooses $s_{u_1}^2$ then her expected latency is:
\begin{align*}
&\underbrace{1+2p+(q_2+1-q_3)(1-p)}_{h_3} + \underbrace{1+2p + q_3(1-p)}_{g_2} 
+\underbrace{1+2p+(1-q_2+q_3)(1-p)}_{h_2}\\
&= 3+ 6p + (1-p)(2+q_3) .
\end{align*}
However,
\begin{align*}
&3+ 6p + (1-p)(2+q_3) - (2+ 4p +(1-p)(3 -2 q_2 - q_3))\\
&= 1 + 2p + (1-p)(-1+2q_2+2q_3)\\
&\ge 3p.
\end{align*}
So $u_1^s$ is always better of by choosing $s_{u_1}^1$. 

By symmetry it follows that for each $p>0$ there is a unique (pure) Bayesian Nash equilibrium $\vsi$
where $\sigma(u^s)=s_u^1$ and $\sigma(u^m)=s_u^3$ for all players $u\in\N$. For its social cost we get
$\SC(\Psi,\vsi)=2+4p$.

So, for each $\epsilon>0$ there is a sufficiently small $p$, such that
\begin{align*}
\wom(\Psi)=\frac{\SC(\Gamma_\Psi,\s)}{\SC(\Psi,\vsi)}=\frac{5}{2+4p}\ge \frac{5}{2}-\epsilon.
\end{align*}
This finishes the proof of the theorem.
\qed
\end{proof}
This is actually a tight result, since for the considered class of malicious Bayesian games the Windfall of
Malice cannot be larger than the Price of Anarchy of the corresponding class of congestion games which 
was shown to be $\frac{5}{2}$ in \cite{CK05}.

\section{Conclusion and Open Problems}
In this paper, we have introduced and studied a new extension to congestion games, that we call malicious Bayesian congestion games. 
More specifically, we have studied problems concerned with the complexity of deciding the existence of
pure Bayesian Nash equilibria. Furthermore, we have presented results on the Price of Malice.

Although we were able to derive multiple interesting results, this work also gives rise to many interesting
open problems. We conclude this paper by stating those, that we consider the most prominent ones.

\begin{itemize}
\item
Our \NP-completeness result in Theorem~\ref{t:npcscg} holds even for linear latency functions, 
identical type probabilities, and if all strategy sets are singleton sets of resources. 
However, if such games are further restricted to symmetric 
games and identical linear latency functions, 
then deciding the existence of a pure Bayesian Nash equilibrium
becomes a trivial task. We believe that this task can also be performed in polynomial time for  
{\em non-identical} linear latency functions and symmetric strategy sets. 
\item
Another, interesting problem in 
this perspective is to reduce the constants in Theorem~\ref{t:npcscg_ex} or show that this is not possible.
\item
Although the upper bound in Corollary~\ref{c:pob} and the corresponding 
lower bound in Theorem~\ref{t:lbi} are asymptotically tight, there is still potential to improve. 
We conjecture that in this case $\pob(\tsl)=\tsl+O(1)$.
\item
We believe that the concept of malicious Bayesian games is very interesting and deserves further
study also in other scenarios. We hope, that our work will encourage others to study such malicious 
Bayesian games.
\end{itemize}

\section{Acknowledgments}
We are very grateful to Christos Papadimitriou and Andreas Maletti for many fruitful discussions on the topic.
Moreover, we thank Florian Schoppmann for his helpful comments on an early version of this paper.

\end{document}